\documentclass[10pt,conference]{IEEEtran}
\usepackage[pdftex]{graphicx}

\usepackage{amsmath}
\usepackage{amsfonts}
\usepackage{amssymb}
\usepackage{setspace}
\usepackage{verbatim}
\usepackage{amsthm}
\usepackage{bbm}
\usepackage{multirow}
\usepackage{psfrag}
\usepackage{booktabs}
\usepackage{xfrac}
\usepackage{subfig}
\usepackage[active]{srcltx}

\newtheorem{thm}{Theorem}

\newtheorem{lem}[thm]{Lemma}

\newtheorem{algorithm}{Algorithm}

\DeclareMathSizes{10}{9}{8}{6}

\IEEEoverridecommandlockouts

\begin{document}

\title{Energy-Delay Considerations in Coded\\Packet Flows}

\author{
\IEEEauthorblockN{Daniel E.~Lucani}
\IEEEauthorblockA{Instituto de Telecomunica\c c{\~o}es, DEEC\\
Faculdade de Engenharia, Universidade do Porto, Portugal\\
dlucani@fe.up.pt}\\
\and
\IEEEauthorblockN{J{\"o}rg Kliewer}
\IEEEauthorblockA{Klipsch School of Electrical and Computer Engineering\\
New Mexico State University, NM, USA\\
jkliewer@nmsu.edu}
\thanks{This work was supported in part by the European Commission under grant
FP7-INFSO-ICT-215252 (N-Crave Project) and by the U.S.~National
Science Foundation under grant CCF-0830666.}
}
\IEEEaftertitletext{\vspace{-2\baselineskip}}

\maketitle

\begin{abstract}
We consider a line of terminals which is connected by
packet erasure channels and where random linear network coding is carried out
at each node prior to transmission. In particular, we address an online
approach in which each terminal has local information to be conveyed
to the base station at the end of the line and provide a queueing
theoretic analysis of this scenario. First, a genie-aided
scenario is considered and the average delay and average
transmission energy  depending on the link erasure probabilities and the
Poisson arrival rates at each node are analyzed. We then assume that all
nodes cannot send and receive at the same time. The transmitting
nodes  in the network send coded data packets before stopping to
wait for the receiving  nodes to acknowledge the number of
degrees of freedom, if any, that are required to decode correctly the
information. We analyze this problem for an infinite queue size at the
terminals and show that there is an optimal number of coded data
packets at each node, in terms of average completion time or transmission energy, to be sent
before stopping to listen. 
\end{abstract}

\vspace{-2ex}
\section{Introduction}
\vspace{-0.5ex}
In networks, the transfer of packets from source to destination can be in
general modeled as a flow of packets \cite{FF62,ACLY00}. Such a flow
is typically routed through intermediate nodes in which the packets are stored
in buffers for subsequent transmission. Further, other flows may join existing
flows at intermediate nodes in order to be routed towards the same direction
downstream in the network. Of particular interest in these scenarios is the
average end-to-end delay of packets associated with a specific flow. At the
same time, in many scenarios related to networks with
energy-constrained nodes,
the average completion energy
of conveying a packet from source to destination is required to be as small as
possible. Satisfying these constraints is particularly challenging in wireless
networks where the physical links between nodes may become unreliable due to
noise, interference, and fading due to node mobility, which typically leads to
packet erasures. 

For such packet-erasure networks, one approach to reliable
transmission is to employ 
random linear
network coding \cite{CWJ03,HMKKESL06} over stored packets at
each node. In the following we will model packet flows in networks by
simple erasure line networks. For such networks, it has been shown in
\cite{PFS05,LPFMK06} that in-network coding is beneficial
compared to a traditional end-to-end forward erasure correction
approach, and that the min-cut capacity can be achieved
asymptotically. The expected delay for multihop line networks and
random linear coding has been characterized in \cite{DDHE09},
and in \cite{VRF07,TVF10} a queueing-theoretic
analysis of finite buffer effects has been carried out.

In this paper we consider packet flows in two-hop erasure
line networks. As a new result we study the practically important case
of multiple flows by assuming that \emph{both} the first and the
second node in the line have local information packets with
Poisson-distributed arrivals available which are demanded by the
receiver. For a related scenario and deterministic channels between
the nodes the capacity region has been recently established in
\cite{LKH10}. In our work, we address both online and batch-to-batch
approaches and provide a queueing-theoretic analysis, where average
delay and average energy consumption as a function of the link erasure
probabilities and the arrival rates at each node are analyzed. We then
assume that all terminals cannot send and receive at the same time, 
which is an extension of the results for the point-to-point case
\cite{LSM09,SJ09} to multihop networks. 
We show that there is an optimal number of coded data
packets at each node, for example in terms of average completion time
or energy, to be sent before stopping to listen, and devise an
efficient algorithm to find these values. Finally, we compare our half-duplex
 schemes with selective repeat (i.e., a scheme with no coding).

\vspace{-1ex}
\section{Genie-Aided Inter-Session Coding}\label{sec:inter}
Let us assume a line network with three nodes, where two adjacent
nodes ($S_1$, $S_2$) are source nodes and the final node is the
destination, $R$ (see Fig.~\ref{fig:system}).  Each source $S_k$
generates data packets at a rate of $\lambda_k$ via a Poisson process;
we assume that the packet arrivals at each source are independent from
each other. This defines the following packet flows, $S_1\rightarrow
S_2\rightarrow R$ and $S_2\rightarrow R$, where flow $k$ is the flow
originating at $S_k$.  We consider an online approach where input
packets arrive continuously. 
Further,
our initial system model has normalized slotted time where parallel transmission
channels are assumed and thus node $S_2$ operates in full-duplex mode.
Therefore, at most one packet can be transmitted from $S_1$ to $S_2$
and at most one from $S_2$ to $R$ per slot,
where $p_1$ and $p_2$ are the corresponding erasure probabilities on
the links between $S_1$ and $S_2$ and $S_2$ and $R$, respectively.
Thus, in order to ensure stability for the queues we assume that
$\lambda_1 \!<\! 1\!-\!p_1$ for $S_1$ and $\lambda_1+\lambda_2<1\!-\!p_2$ for $S_2$.   
Each source node performs inter-session
random linear network coding, where at $S_2$ all incoming flows are
linearly combined. We also assume that each node in the network has
full system knowledge provided by a genie.

As in previous works (see, e.g., \cite{LPFMK06,VRF07,TVF10}) we model the
system as a Markov process. A state ${\cal S} = (i_1,i_2)$ is defined by $i_1$
(or $i_2$) which denotes the number of degrees of freedom (dof) at $S_1$ (or
$S_2$) that have not been seen at $S_2$ (or $R$). The state variable $i_k$
represents the number of (coded) packets in the queue $S_k$ because all
remaining packets can discarded from the queue \cite{KSM08}.  We define
$a_{(x_1,x_2)} \{ b\}$ as the probability of $x_i$ packets being generated at
$S_k$ in $b$ time slots, $k = 1,2$.
Given independence we obtain
$a_{(x_1,x_2)} \{ b\}= a^{(1)}_{(x_1)} \{ b\} a^{(2)}_{(x_2)}\{ b\}$,
where $a^{(i)}_{(x_i)} \{ b\}= \frac{ e^{-\lambda_i  b} {\left(
      \lambda_i b\right)}^{x_i} }{x_i !}$. Further, let
$d_{(y_1,y_2) | {\cal S}} \{ b_1, b_2\}$ be the probability of $y_k$
packets being transmitted successfully from $S_k$ conditioned on the
current state ${\cal S}$ when $b_k$ coded packets, generated from the
$i_k$ packets in the queue, are transmitted. Since we have parallel
transmission channels, $d_{(y_1,y_2) | {\cal S}} \{ b_1, b_2\}=
d_{(y_1) | i} \{ b_1\}d_{(y_2) | j} \{ b_2\}$, where
\vspace{-1.5ex}
\begin{multline*}
d_{(y_k) | x} \{ b\} = \\
\begin{cases}
1&\text{if }x = 0, y_k = 0,\\
\binom{b}{y_k} \left( 1 - p_k\right)^{ y_k } p_k^{b-y_k}  &\text{if } x > 0, y_k = 0,.., \min (x\!,b\!)-\!1\\
\sum_{m = x} ^{b } \binom{b}{m} \left( 1 - p_k\right)^{ m } p_k^{b-m}  &\text{if }x > 0, y_k = \min( x, b) \\
0&\text{otherwise.}
\end{cases}  
\end{multline*} 
Herein, $x$ denotes the number of state transitions to reach the zero state.

Let us further define $P(T|{\cal S}) = P_{{\cal S} \rightarrow {\cal
    S}'}$ as the transition probability between states ${\cal S} =
(i_1,i_2)$ and ${\cal S}' = (i_1',i_2')$. This effect is captured by
the probability of the random vector $T = ( \Delta_1, \Delta_2)$,
where $\Delta_k = i_k' - i_k$. Thus, the transition probability
between states ${\cal S} = (i_1,i_2)$ and ${\cal S}' = (i_1',i_2')$
can be written as
\begin{multline*}
P ( \Delta_1, \Delta_2  |{\cal S} ) = \\
\sum_{y_1\in\{0,1\}, y_2\in\{0,1\}} 
 a_{(\Delta_1,\Delta_2 - f(y_1,y_2) )}\{ 1\} \,d_{(y_1,y_2) | {\cal S}} \{ 1,1\}
\end{multline*}
where
\begin{eqnarray}
f(y_1,y_2) =
\begin{cases}
0&\text{if }y_2 = 0, y_1 = 0, \text{ or if } y_2 = 1, y_1 = 1,\\
-1&\text{if }y_2 = 1, y_1 = 0,\\
1&\text{if }y_2 = 0, y_1 = 1.
\end{cases} 
\notag
\end{eqnarray} 
After some intermediate steps, $P ( \Delta_1, \Delta_2  |{\cal S} )$ 
can be written as
\begin{multline}
P ( \Delta_1, \Delta_2  |{\cal S} ) = 
d_{(0) | i_1}  \{ 1\} \, a^{(1)}_{(\Delta_1)} \{ 1\} \,\mathbf{1}_{ \{
  \Delta_1 \geq 0 \} } \cdot\\
\Big[ d_{(0) | i_2} \{ 1\} \, a^{(2)}_{(\Delta_2)}\{ 1\} \, \mathbf{1}_{ \{ \Delta_2 \geq 0 \} }   + 
 \\
d_{(1) | i_2} \{ 1\}\, a^{(2)}_{(\Delta_2 + 1)} \{ 1\} \,\mathbf{1}_{ \{ \Delta_2  \geq -1 \} } \Big] +
\\
d_{(1) | i_1}\{ 1\} \, a^{(1)}_{(\Delta_1 + 1)} \{ 1\} \,\mathbf{1}_{ \{
  \Delta_1  \geq -1 \} } \cdot \\ 
\Big[ d_{(0) | i_2} \{ 1\} \, a^{(2)}_{(\Delta_2 -1)} \{ 1\} \, \mathbf{1}_{ \{ \Delta_2 \geq 1 \} }  + 
\\
d_{(1) | i_2}\{ 1\} \,  a^{(2)}_{(\Delta_2)}\{ 1\} \, \mathbf{1}_{ \{  \Delta_2 \geq 0 \} } \Big]\label{Eq_TransitionProb}
\end{multline}
where $\mathbf{1}_{ \{ s \in S \} }$ denotes the indicator function
being one when $s \in S $ and zero otherwise.

\begin{figure}[t]
\centerline{\includegraphics[scale=0.5]{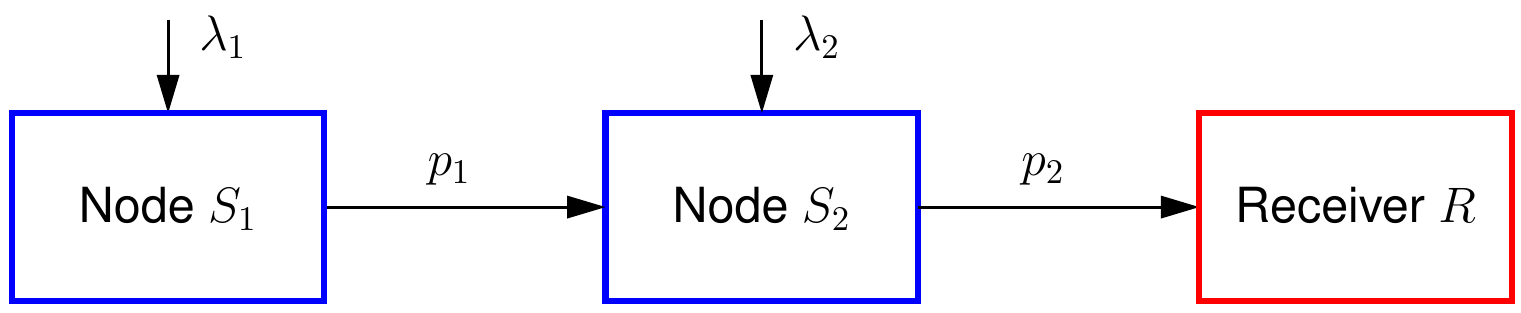}}
\vspace{-1ex}
    \caption{System setup.}
\vspace{-4ex}
  \label{fig:system}
\end{figure}

\subsection{Probability Generating Function}
\label{sec:PGF}
In the following, we consider the probability generating function (PGF) for
the state transition probabilities $P(T|{\cal S})$. The PGF is useful
for computing the steady-state distribution of the underlying Markov
process as discussed below. 
The PGF for a random
vector $X = \{ x_1, x_2, \dots, x_n \}$ is defined as
\begin{equation}
M_{X} (Z) = \sum_ K P(X=K) \prod_i z_i^{k_i} 
\end{equation}
where $Z= \{ z_1, z_2,\dots, z_n \}$ and  $K= \{ k_1, k_2,\dots, k_n \}$. Clearly,
\begin{multline}
\frac{\partial ^{k_1}}{\partial z_1 ^{k_1}} \dots \frac{\partial
  ^{k_n}}{\partial z_n ^{k_n}} M_{X} (Z) \Big|_{z_1 = 0,\dots, z_n =
  0} \\= P(x_1=k_1,\dots, x_n = k_n),
\end{multline}
which simplifies the computation of the individual transition
probabilities in our approach.  For our system we define $M_{T| {\cal
    S} } (Z)$ as the PGF for the state transition
probability when starting in state ${\cal S}$. We have the following lemma.

\begin{lem} \label{thm:GenieIntersession}
  Let $P(T|{\cal S}) = P_{{\cal S} \rightarrow {\cal S}'}$ be the
  transition probability between ${\cal S} = (i_1,i_2)$ and ${\cal S}' =
  (i_1',i_2')$, where $T=(i'_1-i_1,i'_2-i_2)$. The PGF for the genie-aided case with
  inter-session coding is given as
\begin{multline*}
M_{T| {\cal S} } (z) =e^{\lambda_1 (z_1 - 1)}e^{\lambda_2
  (z_2 - 1)}\cdot \\ \left( d_{(0) | i_2}\{ 1\} + d_{(1) | i_2 } \{ 1\}
  z_2^{-1} \right) \cdot \\ 
\left( d_{(0) | i_1}\{ 1\}  + d_{(1) | i_1} \{ 1\}
  z_2 z_1^{-1} \right).
\end{multline*}
\end{lem}
The proof is omitted due to space constraints.

Positive recurrence of the Markov chain can be shown
by using the criteria in \cite{R80}, which guarantees existence of 
a unique stationary probability. 
Let us then define $\pi_{(m,l)}$ as the stationary probability of state
$(m,l)$. 
 Using the fact that 
\begin{equation}
\pi_{(m,l)} = \sum_{m' \geq 0, l' \geq
  0} \pi_{(m',l')} P_{(m',l') \rightarrow (m,l)}
\label{eq:pi_1}
\end{equation}
 the PGF for
$\pi_{(m,l)}$ can be written as $ \Pi(z_1,z_2) = \sum_{m\geq 0,l \geq
  0} \pi_{(m,l)} M_{T| (m,l) }(Z)$.  Given that each node can process
at most one packet per time slot, the PGF has only four cases of
interest: three of them correspond to one or both queues being empty,
and the fourth corresponds to the scenario in which both queues have
at least one packet. The latter translates into $M_{T| (1,1) } (Z) =
M_{T| (a,b) } (Z)$ for $a\geq 1, b\geq 1 $. We exploit this fact to
express the $\Pi(z_1,z_2)$ as
\begin{multline}
\Pi(z_1,z_2) (M_{T| (1,1) } (Z) - z_1 z_2) = \\
   \sum_{ g_1 + g_2 \leq 1} \pi_{(g_1,g_2)} z_1^{g_1} z_2^{g_2} - z_1 z_2 \pi_{(0,0)} M_{T| (0,0) } (Z) +\\
 z_1 z_2 \sum_{g_1\geq 1} \pi_{(g_1,0)} M_{T| (1,0) }
 (Z) + 
 z_1 z_2 \sum_{g_2\geq 1} \pi_{(0,g_2)}  M_{T| (0,1) } (Z) . \notag
\end{multline}

This provides an expression for $\Pi(z_1,z_2)$ in terms of its
coefficients $\pi_{(0,0)}, \pi_{(0,1)},\pi_{(1,0)}, \pi_{(1,1)}$.
Searching for the roots of $M_{T| (1,1) } (Z) - z_1 z_2$ allows us to
find linear equations in terms of the unknown coefficients by
evaluating the above expression with the obtained roots. Further, from
\eqref{eq:pi_1} we obtain directly that $\pi_{(0,0)} = \pi_{(0,0)} P (
0, 0 |(0,0) ) + \pi_{(0,1)} P ( 0, -1 |(0,1) ) $ and $\pi_{(1,0)} =
\pi_{(1,0)} P ( 0, 0 |(1,0) ) + \pi_{(0,1)} P ( 1, -1 |(0,1) ) +
\pi_{(0,0)} P ( 1, 0 |(0,0) ) + \pi_{(1,1)} P ( 0, -1 |(1,1) )$, which
can be used to determine enough linear equations to solve for the
unknown variables $\pi_{(0,0)}, \pi_{(0,1)},\pi_{(1,0)}, \pi_{(1,1)}$.

\subsection{Delay}
\label{sec:intra_delay}
Let us define $D_k$ as the time that a packet in $S_k$ experiences
between being received and being seen at the next hop, and thus
discarded from the queue of $S_k$.
By Little's Law we obtain
\begin{align}
  E[D_1] &= \frac{E[i_1]}{\lambda_1}  = \frac{  \sum_{m\geq 0} m \big(
      \sum_{l\geq 0}  \pi_{(m,l)} \big)   }{\lambda_1},\label{eq:delay1}\\
  E[D_2] &= \frac{E[i_2]}{\lambda_2 + \lambda_1} = \frac{  \sum_{l\geq
      0} l \big(  \sum_{m\geq 0}  \pi_{(m,l)} \big)   }{(\lambda_1 +
    \lambda_2)}.
\label{eq:delay2}
\end{align}
Note that a packet from flow $1$ will experience an average delay of
$E[D_1] + E[D_2]$ before been see at the end receiver $R$, while a
packet from flow $2$ will experience an average delay of $E[D_2]$
before being seen at $R$.

\subsection{Energy}
\label{sec:intra_energy}
We study the average total energy invested per successfully
transmitted packet for each of the two transmitting nodes, $S_1$ and
$S_2$. We consider $E_k$ to be the overall energy to convey a packet
over a time slot $S_k$ (including transmission and reception energy).
Each source is considered to operate in cycles, where each cycle has
two phases. First, we have an ``idle'' phase where the queue for $S_k$ is empty,
which requires $T^0_k$ time slots. Second, there is a ``busy'' phase where the queue
is not-empty, which requires $T_k$ time slots. This constitutes the time
the system needs to obtain $i_k = 0$ for the first time, given that the system
starts at $i_k>0$ after the reception of packets at the end of the
previous idle phase.
\begin{thm} \label{thm:EnergyGenie}
  The average overall energy per transmitted packet at node $S_k$,
  ${\cal E}_k$, for the genie-aided case with inter-session coding is
  given by
  \[ {\cal E}_1 = (1 - P_{em})\, \frac{ E_1 }{\lambda_1},\quad
{\cal E}_2 = (1- P_{em})\,\frac{ E_2 }{ (\lambda_1 + \lambda_2)}, 
\]
where 
\[
E[T^0_1] = \frac{1}{1 - e^{- \lambda_1}}, \  E[T^0_2] = \frac{1}{1 -
  e^{- \lambda_1 } (P_{em} + p_1 (1-P_{em}))},
\]
and $P_{em}= \frac{E[T^0_1] } {E[T_1] + E[T^0_1]}$.
\end{thm}

\begin{proof}
  We present the proof for ${\cal E}_1$. The case of ${\cal E}_2$
  follows naturally. 
  For the genie-aided case, at most one packet can be in the server at any
  time. When the system is empty, the source will not transmit and no energy
  is invested in this process.  
  The PASTA-property \cite{Wolff82} implies that the probability $P_{em}$ that
  a packet arrives at an empty system is given by the probability that the
  system is empty at an arbitrary time.  Using a standard argument from
  renewal theory, the probability of a system being empty is given by the mean
  idle time divided by the mean cycle time, i.e., $P_{em} =  \frac{E[T^0_1]  }
  {E[T_1] + E[T^0_1]}$. Then, the mean energy per time slot is given by $E_1 (1-P_{em})$.
Dividing this by the arrival rate per time slot $\lambda_1 $
yields the energy invested for transmissions from $S_1$.  Since all
incoming packets of the first source are due to Poisson arrivals,
$E[T^0_1] = 1/(1 - e^{- \lambda_1 })$.  For computing $E[T^0_2]$
it is necessary to consider two sources of incoming packets: the ones
which are locally generated and the ones which are received from
upstream in the network. Using $P_{em}$, it is clear that
$E[T^0_2] = 1/(1 - e^{- \lambda_1 } (P_{em} + p_1 (1-P_{em})))$.
The rest of the proof follows naturally.
\end{proof}

\vspace{-0.5ex}
\section{Genie-Aided Intra-Session Coding}
\vspace{-0.5ex}
In this case we separate both flows by performing random linear coding
only within a single flow. Thus, the state representation from
Section~\ref{sec:inter} needs to be extended by another state variable
$i_3$. In particular, the new state is defined as ${\cal L} =
(i_1,i_2,i_3)$, where $i_1$ represents the dof present at $S_1$ that
have not been seen at $S_2$ from flow 1, $i_2$ represents the dof
present at $S_2$ that have not been seen at $R$ from flow 1, and $i_3$
represents the dof present at $S_2$ that have not been seen at $R$
from flow 2. Since we can only service one packet per time slot and
since $S_2$ must hence choose one flow for servicing at each time slot, let us
first describe the transition probability conditioned on the flow that
has been chosen for service. This allows to model different
scheduling or resource allocation mechanisms which are  implemented at node
$S_2$ in order to serve both flows.

The probability of transition from state ${\cal L}$ to state ${\cal
  L}'$ is given as ${\cal P}( T | {\cal L} )$.  We define the event
$A_i$ as the event of flow $i$ being serviced during the current time
slot by node $S_2$. Define $\Delta_k = i_k' - i_k$.  First, let us
consider the case in which we condition on flow 1 being serviced,
i.e., 
\begin{equation}
{\cal P}( T | {\cal L} , A_1 ) =  a^{(2)}_{(\Delta_3)} \{ 1\} \left [
  P( T | {\cal S} ) |_{\lambda_2 = 0} \right]
\label{eq:P_A1}
\end{equation}
where $P( T | {\cal S} ) |_{\lambda_2 = 0} $ is the state transition
probability defined in \eqref{Eq_TransitionProb} and evaluated for the case of $\lambda_2 =
0$. If we condition  ${\cal P}( T | {\cal L} )$ on event $A_2$ we obtain
\begin{multline}
{\cal P}( T | {\cal L} , A_2 ) = \Big( d_{(1) | i_3} \{ 1\} a^{(2)}_{(\Delta_3 + 1)} \{ 1\} \mathbf{1}_{ \{ \Delta_3 + 1 \geq 0 \} } + \\
 d_{(0) | i_3}\{ 1\} a^{(2)}_{(\Delta_3)} \{ 1\} \mathbf{1}_{ \{ \Delta_3 \geq 0 \} } \Big) \cdot \\
\Big(   d_{(1) | i_1}\{ 1\} a^{(1)}_{(\Delta_1 + 1)} \{ 1\}
\mathbf{1}_{ \{ \Delta_1 + 1 \geq 0 \} } \mathbf{1}_{ \{  \Delta_2 = 1
  \} }  + \\
 d_{(0) | i_1}\{ 1\}a^{(1)}_{(\Delta_1)}\{ 1\} \mathbf{1}_{ \{
   \Delta_1 \geq 0 \} } \mathbf{1}_{ \{  \Delta_2 = 0 \} } \Big).
\label{eq:P_A2}
\end{multline}
Note that if node $S_2$ implements a policy for choosing the serviced flow in
terms of the state, either $A_1$ or $A_2$ will happen depending on
${\cal L}$. If the system uses a randomized policy, e.g., if it chooses
event $A_1$ with scheduling probability $P_s$  regardless of the
starting state ${\cal L}$, then the overall transition probability
will be obtained as ${\cal P}( T | {\cal L} )=P_s {\cal P}( T | {\cal
  L} , A_1 )+(1-P_s) {\cal P}( T | {\cal L} , A_2 )$.

Let us define $D_k^h$ as the time that a packet in $S_k$ from flow $h$
experiences between being received and being seen at the next hop, and
thus discarded from the queue of $S_k$. We can define the expected
delay analogous to Section~\ref{sec:intra_delay} for the intra-session
case. Note that we may devise the policy for servicing
flow $1$ and flow $2$ in such a way that $E[D_1^1] + E[D_2^1] \approx
E[D_2^2]$, thus providing delay fairness to both flows. 

Likewise,
similar considerations as in  Section~\ref{sec:intra_energy} for the
overall transmission energy also apply in the inter-session case.



\vspace{-0.5ex}
\section{Half-Duplex Inter-Session Case} \label{HD_IS_Online}\
\vspace{-0.5ex}
We now introduce a half-duplex constraint on the
problem in the sense that node $S_2$ can only transmit or receive packets,
but not both, in a single time slot. We also assume that each node has
only access to local information and that ACK packets are employed to
update the knowledge about the state of other nodes in the network.
ACK packets introduce additional delay and energy consumption. We further
assume that $S_1$ receives acknowledgments piggybacked in the header
of the transmission packets from intended to be sent from $S_2$ to
$R$.

Let us consider the state $(i_1,i_2,S_t)$, where $i_1$ and $i_2$
represent the dof missing at node $S_2$ and $R$, resp., and
$S_t$ indicates the node that will be actively transmitting in the
upcoming round. We consider that node $S_1$ can transmit $N_{i_1}$
coded packets in its turn, and that $S_2$ can transmit $N_{(i_1,i_2)}$
coded packets when it has the opportunity to transmit. We also define
a sliding coding window with a maximum number of packets $W_k$ that are part
of a random linear combination for each node $S_k$.  Then, the transition
probability $(i_1,i_2,S_1)$ to state $(i_1',i_2', S_2)$ can be derived
as 
\begin{multline}
P(  \Delta_1, \Delta_2   |(i_1,i_2,S_1)) =
 \sum_{m = 0}^{i_1} a( {\scriptstyle \Delta_1 + m}) \{ {\scriptstyle N_{i_1}}\} a(  {\scriptstyle \Delta_2 - m}) \{
  {\scriptstyle N_{i_1}} \}\cdot \\
\mathbf{1}_{ \{  \Delta_2 -m\geq 0 \} }\mathbf{1}_{ \{  \Delta_1 + m\geq 0 \} } d_{(m) | i_1 } \{ N_{i_1}\}. \label{Eq_TransitionProbHD}
\end{multline}

\begin{lem}
  Let $P(\Delta_1, \Delta_2 |(i_1,i_2,S_1))$ denote the
  transition probability between the states $(i_1,i_2,S_1)$ and
  $(i_1',i_2',S_2)$. The PGF is given as
\begin{multline*}
M_{ (\Delta_1, \Delta_2 )| (i_1,i_2,S_1) } (Z) =
e^{\lambda_1  N_{i_1} (z_1 -1)} e^{\lambda_2  N_{i_1} (z_2 -
  1)} \cdot \\
\Big[  \sum_{m=0} ^{i_1 -1}  \binom{N_{i_1}}{m} \biggl(\frac{1-p_1}{p_1}\biggr)^{m} p_1^{N_{i_1} } z_1^{-m} z_2^{m} + \\
\sum_{m = i_1} ^{N_{i_1}}  \binom{N_{i_1}}{m} \biggl(\frac{1-p_1}{p_1}\biggr)^{m} p_1^{N_{i_1} } z_1^{-i_1} z_2^{i_2}  \Big].
\end{multline*}
\end{lem}
The proof is similar to that of Lemma~\ref{thm:GenieIntersession} and
is also omitted due to space constraints. A similar approach can be
followed for the case of transition from state $(i_1,i_2,S_2)$ to
state $(i_1',i_2', S_1)$. Finally, we define $T^{(i_1,i_2,S_t)}$ as
the time associated to a transition starting at state $(i_1,i_2,S_t)$.
As stated in Section~\ref{sec:PGF} for the genie-aided case we can use the
PGF in the same way to derive expressions for both expected delay and expected energy
consumption.

%
%

\vspace{-0.5ex}
\section{Half-Duplex Inter-Session Coding: Batch-by-Batch}
\vspace{-0.5ex}
Finding the optimal
$N_{i_1}$ and $N_{(i_1,i_2)}$ for the online case discussed in the
previous  section requires an integer search due to
dynamic nature of the online approach. This motivates considering a
batch-by-batch approach where the fact that the Markov chain has an
absorbing state significantly simplifies the complexity of the
optimization problem, as we will describe in the following.

For this case, we model the service process as an absorbing Markov
chain with transition probabilities similar to those defined in
\eqref{Eq_TransitionProbHD} using $\lambda_1=\lambda_2 = 0$. As in the
online case, we consider that there is a coding window with a maximum
size $M_k$ for each node $S_k$. The starting state of the Markov chain is given by
the number of packets in the queue that are passed to the server,
which is limited by the coding window's maximum size. The absorbing
state is constituted by states $(0,0,S_1)$ and $(0,0,S_2)$ in
Section~\ref{HD_IS_Online}. 

We exploit the periodic structure, introduced by the round robin assignment 
of the transmission in our half-duplex scheme, to
estimate $N_{i_1}$ and $N_{(i_1,i_2)}$. Let us define
$T_{(i_1,i_2,S_t)}$ as the mean completion time when the system starts in
state $(i_1,i_2,S_t)$. Note that
\begin{equation}
T_{(i_1,i_2,S_1)} = T^{(i_1,i_2,S_1)} + \sum_{i_1'} P_{(i_1,i_2,S_1) \rightarrow (i_1',i_2,S_2)} T_{(i_1',i_2,S_2)}, \label{TX0_eq.tag}
\end{equation}
\begin{equation}
T_{(i_1,i_2,S_2)} = T^{(i_1,i_2,S_2)} + \sum_{i_2'} P_{(i_1,i_2,S_2) \rightarrow (i_1,i_2',S_1)} T_{(i_1,i_2',S_1)}. \label{TX1_eq.tag} 
\end{equation}
We can substitute \eqref{TX1_eq.tag} into~\eqref{TX0_eq.tag} to obtain
\begin{small}
\begin{multline}
T_{(i_1,i_2,S_1)} =T^{(i_1,i_2,S_1)} + \sum_{i_1'} T^{(i_1',i_2,S_2)} P_{(i_1,i_2,S_1) \rightarrow
  (i_1',i_2,S_2)}\\
 +\sum_{i_1',i_2'} \!T_{(i_1',i_2',S_1)} P_{(i_1,i_2,S_1) \rightarrow
   (i_1',i_2,S_2)}P_{(i_1',i_2,S_2) \rightarrow (i_1',i_2',S_1)}.
\label{eq:T10_9} 
\end{multline}
\end{small}

This expression captures the fact that node $S_1$ can view its
communication channel as a transmission link, which has a random
waiting time between rounds of transmission.  The waiting time depends
on the transmissions of node $S_2$.  We exploit this fact to propose a
search algorithm for finding $N_{i_1}$ and $N_{(i_1,i_2)}$.  A similar
expression can be found for $T_{(i_1,i_2,S_2)}$ and flow 2 by substituting
\eqref{TX0_eq.tag} into \eqref{TX1_eq.tag}, and similar expressions
hold for the analysis of mean energy with small modifications. For our
case, $T^{(i_1,i_2,S_1)} = N_{i_1}  $ and $T^{(i_1,i_2,S_2)} =
N_{(i_1,i_2)}  + 1$. The latter needs to account for a time
slot used for transmitting an ACK from $R$ to $S_2$.

Let us define $\hat N_{i_1} (n)$ (or $\hat N_{(i_1,i_2)} (n)$) as the
estimate for $N_{i_1}$ (or $N_{(i_1,i_2)}$) at step $n$ of the
algorithm, and $\hat P_{(i_1,i_2,S_1) \rightarrow (i_1',i_2,S_2)} (n) $ and
$\hat P_{(i_1,j_2,S_1) \rightarrow (i_1,i_2',S_t)} (n)$ are the transition
probabilities based on the estimates for the $n$-th step.  Finally, we start the algorithm by
setting $\hat N_{i_1}(0) = i_1$, $\hat N_{(i_1,i_2)}(0) = i_2$, $n =
1$.

\begin{algorithm}\label{search_algorithm}


\texttt{S1: TRANSMISSION FROM $S_1$ to $S_2$:}

Compute $\hat N_{i_1}(n), \forall i_1 = 1,\dots,M_1$ to minimize the
completion time of a half-duplex link, as in \cite{LSM09}, using the waiting time  
$
 \sum_{i_1'} (\hat N_{(i_1',M_2)} (n) +1 ) \hat P_{(i_1,M_2,S_1) \rightarrow (i_1',M_2,S_2)}$,
which corresponds to the second term on the right hand side of
\eqref{eq:T10_9}. 

\vspace{1ex}
\noindent
\texttt{S2: TRANSMISSION FROM $S_2$ to $R$:}

\texttt{FOR $i_1' = 1,2,\dots,M_1$}
	
Compute $\hat N_{(i_1',i_2)}(n), \forall i_2 = 1,\dots,M_2$ to minimize the
completion time of a half-duplex link, as in \cite{LSM09}, using  the waiting time  
$
  \sum_{i_2'} \hat 
N_{(i_1',i_2')}(n) \hat P_{(i_1',i_2,S_1) \rightarrow (i_1',i_2',S_1)} (n) 
$.

\vspace{1ex}
\noindent
\texttt{STOPPING CRITERIA:}

\texttt{IF} $\hat N_{(i_1,i_2)} (n) = \hat N_{(i_1,i_2)} (n-1)$, $\forall i,j,t$: 
STOP

\texttt{ELSE} $n\leftarrow n + 1$, and go to S1.

\end{algorithm}

We point out that this search algorithm can be used for both
completion time (as shown above) or with different metrics, such as
energy or the product of energy and delay to reach a trade-off between
both metrics.

\vspace{-0.5ex}
\section{Numerical Results}
\vspace{-0.5ex}
Fig.~\ref{fig:queueingDelay} compares genie-aided inter- and
intra-session coding by illustrating the delay performance for flows
$1$ and $2$, i.e., $E[D_1]$ and $E[D_2]$ in \eqref{eq:delay1} and
\eqref{eq:delay2}. For the intra-session case we use a randomized
policy at $S_2$, which sends packets from flow $1$ with probability
$P_s$ and services flow $2$ otherwise.
For the case of intra-session coding with $P_s = 0.75$, we observe
that the delay for flow $1$ and flow $2$ have the same mean delay
performance at $\lambda_1 = 0.12$. As can be seen from
Fig.~\ref{fig:queueingDelay} the choice of $P_s$ to achieve $E[D_1]
\approx E[D_2]$ depends on the arrival rate.  However, this
illustrates that the combination of intra-session coding and random
scheduling policies at intermediate nodes in a line network can be
used successfully to provide delay fairness for sources at different
number of hops to the final receiver.


\begin{figure}[t]
	\centering
    \includegraphics[width=2.9in,height=2.9in,keepaspectratio]{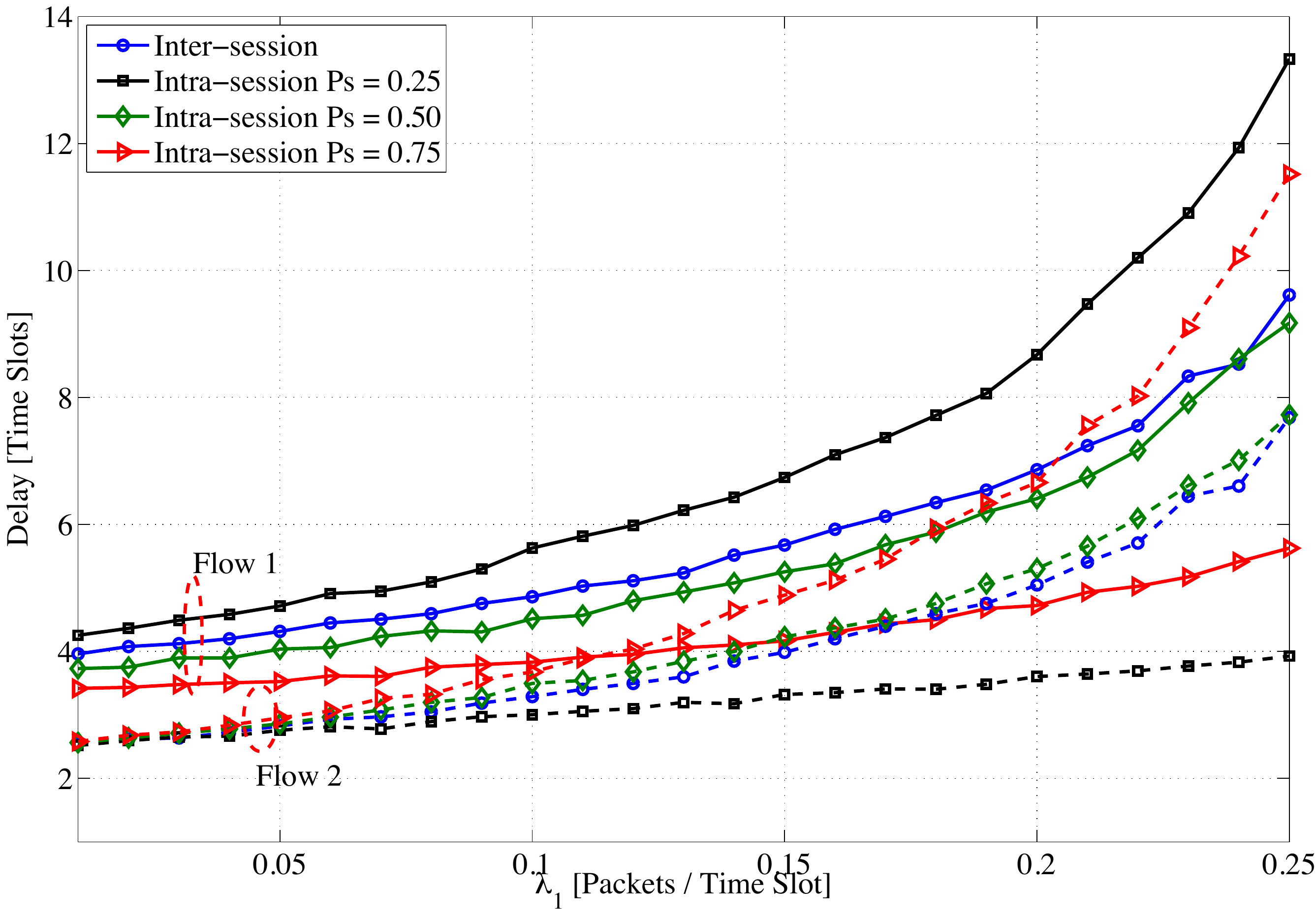}
  \caption{Delay vs arrival rate $\lambda_1$ for genie-aided
    inter- and intra-session coding. Parameters: $\lambda_2 = 0.25$, $p_1 = 0.3$, $p_2 = 0.4$.}
    \label{fig:queueingDelay}
  \end{figure}
  \begin{figure}[t]
	\centering
\includegraphics[scale=0.4]{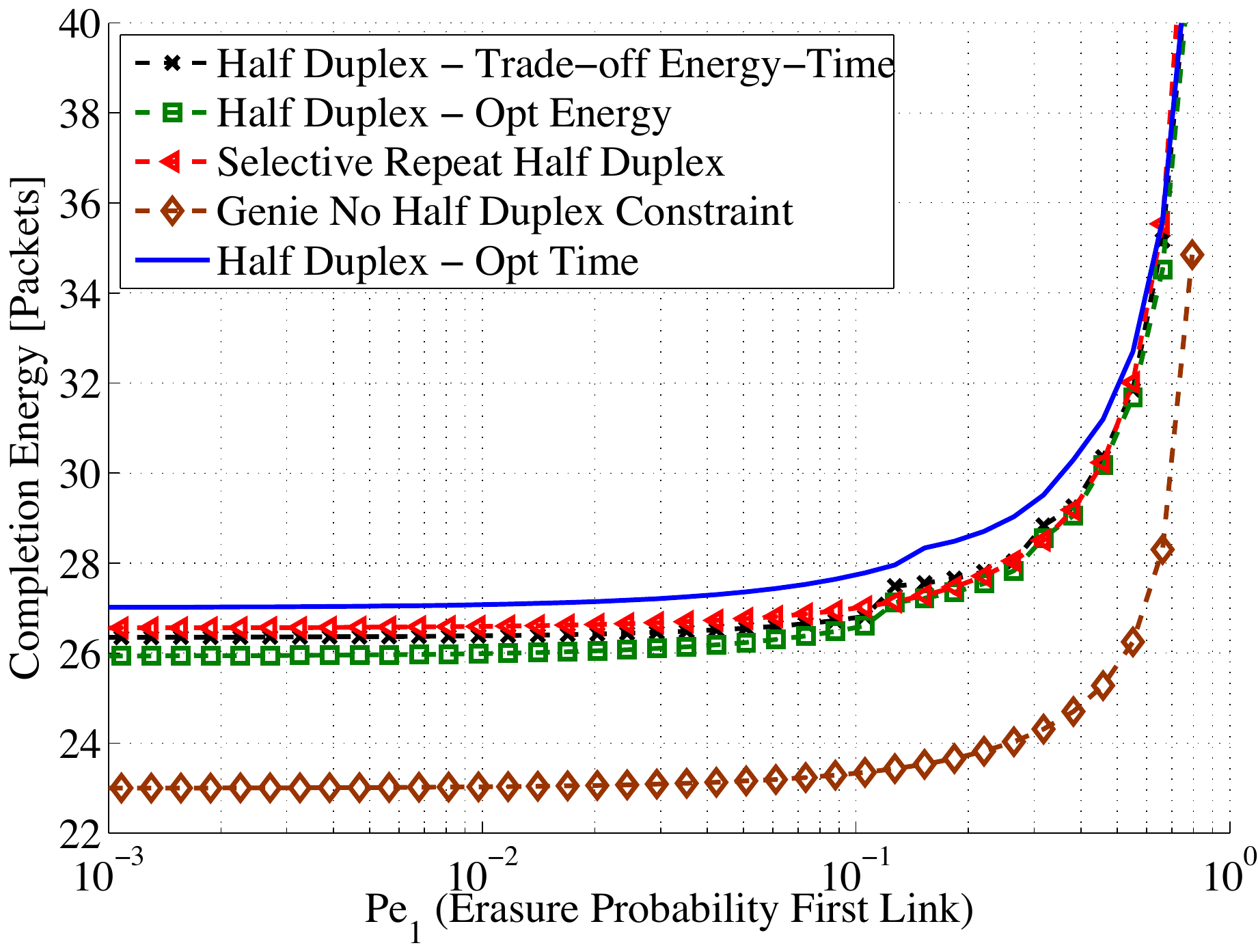}
    \caption{Completion energy in batch-by-batch inter-session coding versus erasure probability.
      Parameters: $p_2 = 0.4$, $M_1 = M_2 =3$ packets, $E_1 = 1$, $E_2 = 2$. }
\vspace{-4ex}
    \label{fig:HD_Energy}
  \end{figure}

  Fig.~\ref{fig:HD_Energy} compares the energy performance for batch-by-batch
  inter-session coding optimized by using Algorithm 1 for different metrics
  and for an uncoded selective repeat (ARQ) strategy.
  Clearly, the completion energy is higher when we optimize for completion
  time (``Opt.~Time'') than optimizing for energy (``Opt.~Energy''). If we use
  a metric that aims to reduce the product of mean energy and and mean time
  simultaneously (``Trade-off Energy-Time''), it yields an intermediate behavior
  between the two.
  The genie case with no half-duplex constraint constitutes a lower bound on
  energy consumption, because i) it does not require ACK packets, and ii) it
  sends only enough to complete the transmission. 

\vspace{-1ex}
\section{Conclusions}
\vspace{-0.5ex}
We have considered a two-hop erasure line network where as an
extension of existing work \emph{each} of the first two nodes intents
to send local information packets with Poisson-distributed arrivals to
the receiver node via random linear network coding. For both online
and batch-to-batch schemes a queueing-theoretic framework based on
Markov chains and the corresponding moment-generating functions for
the state transition probabilities has been provided. We found that
despite intra-session random linear coding is not throughput optimal,
we can achieve delay fairness for flows coming from different sources
by employing a random servicing policy at the middle node in the line
network.  Then, the  half-duplex case is
addressed, and it is shown that there is an optimum number of packets
each node needs to send before stopping to wait for the receiving
nodes to acknowledge the missing degrees of freedom. 
\vspace{-0.6ex}

\bibliographystyle{IEEEtran}

\bibliography{lit}

\end{document}